\documentclass[12pt, twoside]{article}
\usepackage[margin=1in]{geometry}

\usepackage{amsmath, amsthm,amssymb,dsfont}
\usepackage{color,verbatim,enumitem,psfrag,bbm,comment}
\usepackage{algorithm,algpseudocode}
\usepackage{appendix}
\usepackage{xspace}
\usepackage{epstopdf}
\usepackage{epstopdf}
\usepackage{subfig}
\usepackage[demo]{graphicx}
\PassOptionsToPackage{hyphens}{url}
\usepackage{authblk}

\def\EE{{\mathbb{E}}}\def\PP{{\mathbb{P}}}

\def\setR{\mathbb{R}}

\def\bR{\mathbf{R}}
\def\bW{\mathbf{W}}

\newcommand{\fall}{\,\forall\,}

\newcommand{\n}{\nonumber}

\newtheorem{theorem}{Theorem}

\newtheorem{lemma}[theorem]{Lemma}

\newtheorem{definition}{Definition}

\date{}
\title{Detecting Sponsored Recommendations}
\author[1]{Subhashini Krishnasamy}
\author[1]{Rajat Sen}
\author[2]{Sewoong Oh}
\author[1]{Sanjay Shakkottai}
\affil[1]{The University of Texas at Austin}
\affil[2]{University of Illinois at Urbana-Champaign}

\begin{document}
%

\maketitle

\begin{abstract}
  With a vast number of items, web-pages, and news to choose from,
  online services and the customers both benefit tremendously from
  personalized recommender systems.  Such systems however provide
  great opportunities for targeted advertisements, by displaying ads
  alongside genuine recommendations.  We consider a {\em biased}
  recommendation system where such ads are displayed without any tags
  (disguised as genuine recommendations), rendering them
  indistinguishable to a single user.  We ask whether
  it is possible for a small subset of collaborating users to detect
  such a bias.  We propose an algorithm that can detect such a bias
  through statistical analysis on the collaborating users' feedback.
  The algorithm requires only binary information indicating whether a
  user was satisfied with each of the recommended item or not.  This
  makes the algorithm widely appealing to real world issues such as
  identification of search engine bias and pharmaceutical lobbying.
  We prove that the proposed algorithm detects the bias with high
  probability for a broad class of recommendation systems when
  sufficient number of users provide feedback on sufficient number of
  recommendations.  We provide extensive simulations with real data sets and
  practical recommender systems, which confirm the trade offs in the
  theoretical guarantees.

\end{abstract}

%
%

\section{Introduction}
\label{sec:intro}


The growth of online services has provided a vast variety of choices
to users. This choice exists today in multiple domains including
e-commerce with a variety of products, and online entertainment
(NetFlix, Pandora). With users having to choose from an overwhelming
set of items, recommender systems have become indispensable in easing
the information overload and search complexity. Recommender systems
are not restricted to retail businesses. A search engine like Google
can be viewed as a recommendation engine that helps users find
relevant information by ranking the search results according to their
search criteria, history and other personal information. Social
networking sites like Twitter and Facebook display Tweets and News
Feed based on users' past behavior and their connections to other
users. News portals like Yahoo!\xspace News also present personalized content
to online news readers.

Personalized recommender systems serve as an attractive platform for
advertisers to reach their targeted consumers. It is now customary to
see ads alongside other genuine recommendations in many of the
websites that provide recommendation services. One can distinguish
these ads from genuine recommendations, for example, by the location
of their placement or by their special tags. But recommendation
engines are not legally obliged to facilitate such distinction and
could possibly serve these ads mixed with genuine recommendations in a
manner that renders them indistinguishable to users. Such a {\em
  biased recommender system} can have far reaching consequences,
including user dissatisfaction with the recommendations
\cite{beel-etal13sponsored-vs-organic}. A recent survey by Facebook
shows that users find sponsored ads mixed with genuine posts in their
News Feed more annoying than the explicit, well-separated ads
\cite{albergotti14facebook-newsfeed}. Social and political
consequences of bias in the context of media and online content have
also been studied \cite{entman07framing-bias, introna00shaping-the-web, epstein14google-democracy}.

Modern recommender systems, in general, consist of two components:
{\em (i)} learn individual preferences from user feedback, and {\em
  (ii)} recommend items to users based on the estimated preferences.
This combination of learning and recommending is bound to be noisy
(the learning phase will {\em explore} individual preferences
typically by presenting ``random'' recommendations), and several
recommendations to users will likely be ineffective. Critically, both
noise and bias manifest as bad recommendations to users. However,
noise is benign and is a consequence of learning, while bias is
systematic and is to be deprecated. Thus, a basic question of interest
to the users of such systems is whether or not such a biased
recommender system can be detected. This is a broad question, and
detecting bias in its most general sense is out of the scope of this
paper. We focus on a detecting a specific type of bias where
recommendation engines {\em systematically favor a few items over
  other better or at least equally good items, contrary to what an
  objective or unbiased system would do.}

It should be noted that, with most service providers being
non-transparent about their recommendation strategies, one cannot hope
to know the exact statistical profile of the recommendation engine
\emph{a priori}. Therefore, the key is to identify the primary
features that can be used to differentiate between the two types
without any \emph{a priori} knowledge about the particulars of the
recommendation strategy. One could, for instance, consider the average
rating or the average number of ineffective recommendations as the
performance measure and make a decision based on a threshold
parameter. However, as we also demonstrate through simulations, such a
basic algorithm based solely on average performance cannot distinguish
between deliberate systematic bias and innocuous random
  errors. This brings us to the key question: {\em Can we
    develop a better method to expose a biased recommender system?}

\subsection{Contributions of this paper}
\label{sec:contrib}

We say a recommendation engine is \emph{biased}, if it systematically
favors a small set of items over other items in the database
irrespective of users' preferences. On the other hand, we say that a
recommendation engine is \emph{objective}, if it satisfies a simple
monotonic property in its recommendations to users -- better suited
items are given higher priority (in a statistical sense). The primary
goal of this paper is then to develop algorithms to answer the
following question: \emph{Can a meaningful distinction be drawn
  between objective and  biased recommendation engines?}


\noindent \textbf{BiAD Algorithm:} We propose an anomaly detection
algorithm that we call \emph{Bi}nary feedback \emph{A}nomaly
\emph{D}etector (\emph{BiAD}), which uses a statistical approach to
identify a \emph{biased} recommendation engine. Under appropriate
conditions on the size of the ad-pool, the aggressiveness of the
biased recommender system, and the number of users/samples, we show
that BiAD correctly (with high probability) distinguishes between
objective and biased recommendation engines.

The algorithm leverages user collaboration, and is based on the
observation that a \emph{biased} system is typically characterized
by the occurrence of a \emph{large
  number of ineffective recommendations in a small set of items}. On
the contrary, giving higher priority to more effective items, as in an
\emph{objective} recommender system, precludes such concentration in a
small set. Notably, since the users are not aware of the set of items,
the BiAD algorithm is adaptive -- as the recommender system learns
users, the users ``learn'' the recommender system. Further, our
algorithm relies only on binary feedback on the effectiveness of the
recommendation. Finally, the BiAD algorithm also works for a large
class of recommender systems since our model does not place any
constraints on the recommendation engine other than mild statistical
conditions. We finally present extensive simulation results that cover
various types of recommender systems and data sets to illustrate the
wide applicability of the algorithm.

\subsection{Related Work}

Following the recent successes of the targeted advertising services,
there have been several empirical studies that investigate the effects
of displaying sponsored content alongside organic content
\cite{beel-etal13sponsored-vs-organic, ghose-yang09search-engine-advertising, tucker12social-advertising}.
  There have also been attempts
to explain such effects through theoretical models
\cite{yang-ghose10organic-sponsored-relation, bergemann-bonatti11targeted-advertising}. In addition, several
researchers have worked on designing systems and algorithms from the
content provider's perspective for revenue maximization through
efficient auction of the ad-space 
\cite{lahaie-etal07sponsored-search-auctions} and from the
advertiser's perspective for effectively reaching the target audience
\cite{rusmevichientong-williamson06algorithm-search-advertising, turpin-katz08system-social-advertising}.
It is empirically shown in \cite{beel-etal13sponsored-vs-organic} that
customers are less likely to select recommendations which are tagged
as ``advertisement'' or ``sponsored'', motivating the advertisers to remove such tags.

Prior work on anomaly detection in recommender systems exists from the
perspective of a recommendation engine as a victim of false
user-profile injections \cite{burke-etal06attack-detection-recosys, mobasher-etal07attack-recosys}. To the best of our knowledge, ours
is the first work that considers the problem from the users'
perspective and proposes a mechanism for detection of bias in
recommendation engines.

\section{System Model}
\label{sec:sys-model}
In this section, we describe our assumptions about the structural properties of \emph{objective} and \emph{biased} recommender systems by the means of a probabilistic model. This model does not include any particulars about the working of the recommendation engine and therefore typifies a broad class of recommender systems. Before we proceed to describe the model in detail, the salient features of this model are listed below:
\begin{itemize}
\item An \emph{objective} recommendation engine has a fairly good estimate of the user preferences.
\item An \emph{objective} recommendation engine follows the monotonic property -- higher preference to higher ranked items.
\item A \emph{biased} recommendation engine systematically gives preference to a small set of items irrespective of users' tastes.
\end{itemize}

\bigskip\noindent{\bf Notation:}  Our notation $O, \Omega, \Theta, o, \omega$ to describe the asymptotics of various parameters with increasing size of the database (total number of items in the database) is according to the standard Landau notation. We say that an event occurs with high probability if the probability of the event tends to $1$ as the size of the database goes to infinity. We use $\mathds{1}\left\{\cdot\right\}$ to represent the indicator function, i.e.,
\begin{align*}
 \mathds{1}\left\{\mathsf{E}\right\} := \begin{cases} 1 & \text{if event $\mathsf{E}$ occurs}, \\
								0 & \text{otherwise}. \end{cases}
\end{align*}
Equality and inequality between random variables always refer to almost sure (with probability $1$) conditions unless otherwise specified. For example, if $X$ and $Y$ are two random variables, then $X=Y$ implies $X=Y \; a.s.$ For any given matrix, $\bR,$ the $u^{th}$ row of $\bR$ is represented by $\bR_u.$

\subsection{User-Item Database}
The recommendation engine recommends products to users from a large database of $m$ items indexed from $1$ to $m.$ A user's opinion about an item is represented by a numerical value that we call the user's \emph{rating} of that item. It should be noted that these ratings are only an implicit representation of true opinions of the users -- higher the rating, better suited is the item for the user. We denote the user-item rating matrix for the entire database by $\bR,$ where rows indicate users and columns indicate items. We introduce a parameter called the \emph{efficacy threshold}, denoted by $\eta$ which is used to represent opinions on a binary scale. We assume that a user is satisfied with a recommendation if the rating of the recommended item is greater than or equal to $\eta.$ We refer to such a recommendation and item as being \emph{effective} for that user. 
\begin{definition}[Effective \& Ineffective]
\label{defn:effective}
 An item $i$ is effective for a user $u$ if the rating of that item by the user, $R_{ui}$ is at least $\eta.$ Similarly, a recommendation is said to be effective for a user if the recommended item is effective. An item or recommendation that is not effective is said to be ineffective.
\end{definition}
 Let $f_u(\eta, [m])$ denote the number of items in the database $[m]$ whose rating is greater than or equal to $\eta$ for user $u.$ In other words, it is the number of effective items in the database for user $u.$ 

Let us define the function $F: \setR \times \setR^m \rightarrow \setR$ as follows:
\begin{align*}
F(r, \bR_u) & := \left\vert\left\{i: R_{ui} \geq r\right\}\right\vert,
\end{align*}
where $R_{ui}$ is the $i^{th}$ element of the $m$-length vector $\bR_u.$ This function is used to find the number of items whose rating exceeds value $r$ for any player $u$ if the ratings of all the items in the database for player $u$ is given by $\bR_u.$ For example, if $\bR_u$ is the row corresponding to player $u$ in the rating matrix, $\bR,$ then $F(\eta, \bR_u)$ is equal to $f_u(\eta, [m]),$ the number of effective items for user $u.$ Similarly, $F(R_{ui}, \bR_u)$ gives the rank of item $i$ for user $u.$ Also note that for any given $\bR_u,$ $F(r, \bR_u)$ is a non-increasing function of $r.$

\subsection{Recommendation Engine}
\label{subsec:reco-engine}
We next describe the behavior of a recommendation engine using a probabilistic model.
Let $\mathds{1}_{ui}(t)$ indicate whether item $i$ has been recommended to user $u$ at time $t,$ i.e.,
\begin{align*}
\mathds{1}_{ui}(t) & := \begin{cases} 1 & \text{if item $i$ is recommended to user $u$ at time $t$}, \\
								0 & \text{otherwise}. \end{cases}
\end{align*}
We make the following assumption about any recommender system: An item that has been recommended to a user once is not recommended to the same user again, i.e., for any user, $u$ and item, $i,$ $\sum_{t=1}^{\infty} \mathds{1}_{ui}(t) \leq 1.$

\subsubsection{Objective Recommendation Engine}
\label{subsubsec:objective}
An \emph{objective} recommendation engine is considered to consist of two components - one is the \emph{learning} strategy which estimates the user-item rating matrix by the means of available feedback from users, and another is the \emph{recommendation} strategy which generates recommendations based on the estimated user preferences. Our model does not specify the details of the learning strategy except requiring that the output of the strategy, that is the estimate of the user-item matrix, be close to the original rating matrix, $\bR.$ Therefore, this model could be applied to a wide class of recommendation engines which estimate users' preferences fairly well. Let the estimate of the rating matrix at time $t$ be denoted by $\hat{\bR}(t) = \left[\hat{R}_{ui}(t)\right].$ This estimate is modeled as the sum of the original rating matrix and an additive noise matrix whose elements are independent across users, items and time. This can be written as $\hat{\bR}(t) = \bR + \boldsymbol{\epsilon}(t),$ where $\boldsymbol{\epsilon}(t) = \left[\epsilon_{ui}(t)\right]$ is the noise matrix and $\epsilon_{ui}(t)$ is independent of $\epsilon_{u'i'}(t')$ for all $u, u', i, i', t, t'.$

The recommendation strategy uses the estimated user-item rating matrix $\hat{\bR}(t)$ to make recommendations at time $t.$ The following model characterizes the behavior of an \emph{objective} recommendation strategy:
\begin{enumerate}
\item \label{item:stoc-mx} Recommendations are made based on a user-item weight matrix, denoted by $\bW(t) = \left[W_{ui}(t)\right].$ This is a stochastic matrix (rows sum to one), which is updated based on the current estimate of the rating matrix,  $\hat{\bR}(t)$.
\item Given the weight matrix, a user is given a recommendation by choosing an item randomly, independent of everything else, with weights given by the row corresponding to the user in the user-item weight matrix.
\item \label{item:monotonic} At any time $t$, the weight matrix $\bW(t)$ satisfies the following \emph{monotonic} property:
 if $i$ and $j$ are two items that have not been shown to user $u$ and the ratings are such that $\hat{R}_{ui}(t) \geq \hat{R}_{uj}(t),$ then the weights satisfy $W_{ui}(t) \geq W_{uj}(t).$
\end{enumerate}

\subsubsection{Biased Recommendation Engine}
\label{subsubsec:biased}
A \emph{biased} recommendation engine marks a small set of items, $\mathcal{A}$ ($\subseteq [m]$) from the item database as \emph{ads}. To make a recommendation to a user, with probability $\gamma,$ independent of everything else, it chooses an item that has not been shown from the ad-pool, $\mathcal{A}.$ And with probability $1 - \gamma,$ it can follow any recommendation algorithm (for example, an \emph{objective} recommendation algorithm). We refer to $\gamma$ as the \emph{bias probability}. Note that the strategy for showing ad items is unspecified except that no item is shown to a user twice. In particular, the engine may even customize its ad recommendations according to users' tastes. As in the case of the complete database, let $f_u(\eta, \mathcal{A})$ denote the number of effective ads in the ad-pool, $\mathcal{A}$ for user, $u.$

\subsection{Discussion of Assumptions}
\label{subsec:assumptions}
Some of the assumptions in the recommender system model above are present only for ease of analysis. We discuss below how they can be relaxed in practical settings. 
\begin{enumerate}
\item It is assumed that, in any recommendation engine, an item once recommended to a user is not recommended to the same user again. This condition is required only to ensure that there are no repeated recommendations of sponsored advertisements that might be effective. Indeed, if all sponsored ad recommendations are effective, it would not be possible to distinguish them from genuine recommendations. This assumption can therefore be relaxed to require sufficient number of ineffective ad recommendations in a \emph{biased} recommender system.
\item The noise in estimation of the user-item rating matrix is assumed to be additive i.i.d.\ noise. This can be replaced by a more general noise model in which the elements of the estimated user-item matrix are independent across users, items and time. The independence assumption is used to model arbitrary errors which are unlikely to skew the estimated matrix in such a way as to give high preference to a small number of ineffective items uniformly across a large subset of users. 
\item We assume that a \emph{biased} recommendation engine decides to show sponsored ads with probability $\gamma$ (bias probability) independent of everything else. This assumption, again, is used only for ease of exposition. It is sufficient to have arbitrary $\gamma$ fraction of the recommendations from the ad-pool, not necessarily chosen at random.
\end{enumerate}

\section{Anomaly Detection Algorithm and Theoretical Results}
\label{sec:theory-results}

In this section, we describe the algorithm for detecting anomalous systems and provide analysis of Type $I$ and Type $II$ errors as in binary hypothesis testing. 

\subsection{Anomaly Detection System}
\label{subsec:anamoly}
The problem is to design a test to detect if a recommendation engine is \emph{biased}. In other words, the test has to decide between the following two hypotheses:
\begin{itemize}
\item $H_1:$ ``The recommendation engine is \emph{biased},'' 
 and 
\item $H_0:$ ``The recommendation engine is not \emph{biased}.'' 
\end{itemize}
It is similar to a hypothesis testing problem except that the statistical distribution for the two hypotheses are not well defined. The only \emph{a priori} knowledge that is assumed is the structure of a \emph{biased} recommendation engine as specified in Section~\ref{subsec:reco-engine}. But the specifics of various parameters in the recommendation engine, such as bias probability $\gamma$ and the ad-pool $\mathcal{A}$ is unknown. As in traditional hypothesis testing problems, we make use of multiple data points obtained from many users who constitute the \emph{anomaly detection system}. The anomaly detection system consists of a set of $n$ \emph{players} which is a subset of the user database in the recommendation system. These players can give accurate binary feedback (effective or ineffective) on the items recommended to them. Without loss of generality, we denote these players as users indexed from $1$ to $n$ in the user database.
\subsection{Algorithm}
\label{subsec:alg}
We now describe an algorithm called \emph{Bi}nary feedback \emph{A}nomaly \emph{D}etector (\emph{BiAD}), that uses the recommendations made to the players and their feedback to decide between one of the two hypotheses. In every \emph{round} of recommendation, each player is recommended an item by the recommendation engine. In round $t$, the algorithm uses the feedback from the players and computes for each item, the total number of players until that round who have been recommended that item and found the item ineffective. This number is denoted by $B_i(t)$ for item $i.$ If the sum of the largest $\hat{A}(t)$ of these numbers among all the items is greater than or equal to a threshold $T(t),$ the recommendation engine is declared to be \emph{biased}. Otherwise, the same procedure is repeated in the next round. If the algorithm does not declare the engine to be \emph{biased} in $Q(m)$ rounds, then the hypothesis that the engine is \emph{biased} is rejected. Here, $\hat{A}(t),$ $T(t)$ and $Q(m)$ in the algorithm are appropriately designed parameters (given by Equations~\ref{eqn:hatA}-\ref{eqn:sum-prob}). The pseudocode for this algorithm is shown in Algorithm~\ref{alg:biad}.
\begin{algorithm}
  \caption{\emph{Bi}nary feedback \emph{A}nomaly \emph{D}etector (\emph{BiAD})}
  \begin{algorithmic}
    \State Initialize $t=1$ (round $1$).
    \While{$t \leq Q(m)$}
    \State Compute $B_i(t) = $ number of players who have rated item $i$ ineffective upto round $t$ for all $i \in [m].$ 
    \State Compute $S(t) = $ sum of the largest $\hat{A}(t)$ among $\left\{B_i(t)\right\}_{i=1}^{m}.$
    \If{$S(t) \geq T(t)$}
    \State Stop and accept $H_1.$
    \Else
    \State $t \gets t+1.$
    \EndIf
    \EndWhile
    \State Stop and reject $H_1.$
    
  \end{algorithmic}
    \label{alg:biad}
\end{algorithm}

As opposed to the basic average test, this algorithm searches for concentration of large number of ineffective items in a small set. Since the number of potential advertisements is unknown, this algorithm makes decisions in real-time as it gets feedback from the players. Larger the size of the ad-pool, larger is the number of feedback samples required to detect a \emph{biased} engine. (The trade off between various parameters is discussed in detail in Section~\ref{sec:discussion}.) Therefore, the algorithm increases the size of the search set with progressing rounds of recommendation. Also, note that the algorithm requires only binary feedback from the players -- whether the recommendations are effective or ineffective, which explains the name of the algorithm.

 The following theorem gives sufficient conditions for good performance of the algorithm. Unlike in general hypothesis testing problems, we define Type $I$ error, which corresponds to false positives, only for \emph{objective} systems. On the other hand, Type $II$ error is used to refer to missed detection in the case of a \emph{biased} system. We do not give any guarantees for the class of recommendation engines that are neither \emph{objective} nor \emph{biased}.
\begin{theorem}
\label{THM:MAIN}
Let the parameters in the detection algorithm, BiAD satisfy the following equations:
\begin{align}
\hat{A}(t) & = t,	\label{eqn:hatA}\\
T(t) & =  \exp \left(1+W\left(\frac{\hat{\beta}(t)}{e}\right)\right)\hat{p}(t),	\label{eqn:T}
\end{align}
where $W(\cdot)$\footnote{For any $z \in \setR,$ $W(z)e^{W(z)} = z.$} represents the Lambert-W or product log function, and 
\begin{align}
\hat{\beta}(t) & = \frac{\left( \hat{A}(t)+c \right)\log m}{\hat{p}(t)} - 1,	\label{eqn:hatbeta}	\\
\hat{p}(t) & = \exp \left(1+W\left(\frac{\beta(t)}{e}\right)\right)p(t),	\label{eqn:hatp}	\\
\beta(t) & = \frac{\left( \hat{A}(t)+c \right) \log m}{p(t)} - 1,	\label{eqn:beta}\\
p(t) & = \max_{\left\{\hat{\mathcal{A}} \subseteq [m]: \left\vert \hat{\mathcal{A}}\right\vert = \hat{A}(t) \right\}} \left\{ \sum_{u=1}^n \sum_{l=1}^t \EE\left[ P^{\hat{\mathcal{A}}}_u(l) \right] \right\},	\label{eqn:p}\\
 P^{\hat{\mathcal{A}}}_u(l) & = \sum_{i \in \hat{\mathcal{A}}} \frac{\mathds{1} \left\{ R_{ui} < \eta \right\}}{F\left(\hat{R}_{ui}(l), \hat{\bR}_u(l)\right) - l + 1},	\label{eqn:sum-prob}
\end{align}

with $c=1/2.$ Then BiAD gives the following guarantees on the error probabilities:
\begin{enumerate}[label=(\Roman{*})]
\item \label{item:type-I} Type $I$ Error:\\ If the recommendation engine is objective, the probability that BiAD declares it to be anomalous is $O(\frac{Q(m)}{\sqrt{m}}).$
\item \label{item:type-II} Type $II$ Error:\\ If the recommendation engine is anomalous with an ad-pool of size $A,$ and if 
\begin{enumerate}
\item \label{item:ad-size} the number of ads, $A \leq Q(m),$ 
\item \label{item:bias-prob} the fraction of recommendations that are ads, i.e., the bias probability $\gamma = \omega\left( \frac{\log m}{n} \right), \, \omega(\frac{p(A)}{nA}),$  and
\item \label{item:num-eff-ads} $\sum_{u=1}^n f_u(\eta, \mathcal{A}) = o(\gamma nA),$ where $f_u(\eta, \mathcal{A})$ is the number of effective ads for user $u,$
\end{enumerate} 
 then the probability that BiAD does not declare the system as anomalous within $A$ rounds is $e^{-\Omega(\gamma n)}.$
\end{enumerate}
\end{theorem}
The proof of this theorem is presented in Section~\ref{sec:proof}.

\section{Discussion}
\label{sec:discussion}
In this section, we discuss 
how the error probabilities depend on the parameters of the problem. 

\subsection{Choice of Threshold}
\label{subsec:thresh}
Note that computation of the threshold function, $T(t)$ as specified in Theorem~\ref{THM:MAIN} (given by Equation~\eqref{eqn:T}) requires knowledge of the noise statistics and also the players' opinions about all the items in the database. More precisely, since $\hat{\bR}(t) = \bR + \boldsymbol{\epsilon}(t),$ computation of $\EE\left[ P^{\hat{\mathcal{A}}}_u(l) \right]$ (see Equation~\eqref{eqn:sum-prob}) requires knowledge of $\bR_u$ and also the distribution of estimation noise, $\boldsymbol{\epsilon}_u(l).$ The noise statistics reflect the accuracy of the learning strategy of the recommendation engine, and it is possible that these statistics are unknown or cannot be estimated. Moreover, it might also be difficult to obtain the players' opinions about all the items in the database. To overcome this difficulty, a practical implementation of the algorithm could use an approximation of the unknown quantity. We now propose one way to compute such an approximation. Note that
\begin{align*}
\EE\left[ P^{\hat{\mathcal{A}}}_u(l) \right] & = \EE\left[ \sum_{i \in \hat{\mathcal{A}}} \frac{\mathds{1} \left\{ R_{ui} < \eta \right\}}{F\left(\hat{R}_{ui}(l), \hat{\bR}_u(l)\right) - l + 1} \right]	\\
& \leq \EE\left[ \sum_{i \in \hat{\mathcal{A}}} \frac{1}{F\Big(\eta+\epsilon_{ui}(l), \bR_u(l)+ \boldsymbol{\epsilon}_u(l)\Big) - l + 1} \right],	\\
\end{align*}
where the inequality follows since $F\Big(r, \bR_u(l)+ \boldsymbol{\epsilon}_u(l)\Big)$ is a non-increasing function of $r.$ We assume that the estimates of the ratings are not skewed in one direction, and therefore the noise has zero mean. Since the noise statistics are unknown, we could approximate the right hand side of the above inequality by substituting the noise term with its mean. With this approximation, the right hand side of the inequality can be substituted with
\begin{align}
\label{eqn:approx-thresh}
\sum_{i \in \hat{\mathcal{A}}} \frac{1}{F\Big(\eta, \bR_u(l)\Big) - l + 1} & = \frac{\hat{A}(t)}{f_u(\eta, [m]) - l + 1},
\end{align}
where $f_u(\eta, [m])$ is the total number of effective items in the database for user $u.$ Depending on the application, it might be relatively easy to estimate this number or at least estimate a lower bound for this number. As an example, one could roughly estimate that for every user there are $\sqrt{m}$ effective items among the $m$ items in the database. We observe in our simulations that a rough estimate of $f_u(\eta, [m])$ is sufficient to obtain good results. 

Note that over-estimation (under-estimation) of $T(t)$ decreases the probability of Type $I$ (Type $II$) error and increases the probability of Type $II$ (Type $I$) error. In other words, the higher the value of $T(t),$ the lower is the probability of Type $I$ error and the higher is the probability of Type $II$ error. Therefore, the risk associated with false positives and missed detection could serve as a guideline for the choice of the threshold function. In our simulations (Section~\ref{sec:num-results}), we propose a practical threshold function that gives a good balance between the two error probabilities for most scenarios.

\subsection{Effect of Parameters on Performance}
Theorem~\ref{THM:MAIN} gives guarantees on the asymptotic performance of \emph{BiAD} as the size of the database grows large. These guarantees depend on various parameters in the algorithm as well as the recommendation engine. From the analytic bounds derived in Theorem~\ref{THM:MAIN}, we analyze in this section the trade off between these parameters to understand the conditions under which the algorithm shows good performance. We see that the theoretical results support our intuitive understanding about the conditions under which a \emph{biased} system can be distinguished from an \emph{objective} system. These results are also corroborated by our simulation results described in Section~\ref{sec:num-results}.

In Section~\ref{subsec:thresh}, we consider the effect of the choice of the threshold parameter on the error probabilities. We now discuss the effect of other parameters.

\bigskip\noindent{\bf Number of Rounds in the Test and Size of the Ad-Pool}.
It is seen (from Result~\ref{item:type-I}) that the upper bound on the probability of Type $I$ error increases with increasing number of rounds. 
This is expected, since it gives more chances to falsely declare a system {\em biased}. 
For Type $I$ error to go to zero as the size of the database goes to infinity, it is sufficient if $Q(m) = o(\sqrt{m}).$

Guarantees for detection of a \emph{biased} engine (Result~\ref{item:type-II}) are dependent on various parameters. One of the conditions is that the ad-pool is not very large (Condition~\ref{item:ad-size}). Specifically, it is sufficient if the size of the ad-pool, $A$ is at most the maximum number of rounds of recommendations, $Q(m)$. Therefore, increasing $Q(m)$ (the number of rounds of testing) enables detection of larger ad-pools but also increases the probability of Type $I$ error. Intuitively, a small ad-pool conforms with our definition of a \emph{biased} recommendation engine as one that favors a few items over many others and therefore facilitates easier detection.

\bigskip\noindent{\bf Number of Effective Ads}.
For correct detection of a \emph{biased} engine, it is also required that the average number of effective ads (averaged over all players) is not very large (Condition~\ref{item:num-eff-ads}). A large number of effective ads enables the recommendation system to customize ads according to users' tastes and is contradictory to our interpretation of a \emph{biased} system which recommends ads that do not match with users' preferences.

\bigskip\noindent{\bf Number of Players}. 
The dependence on the number of players $n$ is seen in two respects -- it determines the minimum bias probability at which detection is guaranteed (Condition~\ref{item:bias-prob}) and also the probability of Type $II$ error. Both these results show that a large number of players improves the prospect of correct identification which can explained by the fact that a large sample size supports better statistical analysis.

\bigskip\noindent{\bf Number of Effective Items}.
 The minimum bias probability at which detection is ensured is also determined by the average number of effective items in the entire database. This can be seen from the term $\frac{p(A)}{nA}$ in (Condition~\ref{item:bias-prob}). The no estimation noise case ($\epsilon_{ui}(t) = 0$ for all $u, i, t$) is useful in understanding the term $\frac{p(A)}{nA}.$ When there is no noise, $\frac{p(A)}{nA} = O\left(  \frac{1}{nA}\sum_{u=1}^n \sum_{l=1}^A \frac{A}{f_u(\eta,[m])-l+1}\right).$ 
 Therefore, $\frac{p(A)}{nA}$ has an inverse relation with the number of effective items in the database. This conveys that a large number of effective items facilitates better detection of a \emph{biased} engine. Intuitively, a large number of effective items in the database helps in clearer demarcation of an \emph{objective} engine from a \emph{biased} one. With many effective items in the database, an \emph{objective} system would have a higher probability of recommending effective items, while a \emph{biased} system always makes at least $\gamma$ fraction of its recommendations from the ad pool where the number of effective ads is limited. 
 
\bigskip\noindent{\bf Bias Probability.}
The more a \emph{biased} engine recommends from the ad-pool, the more apparent is its biased behavior. The fraction of total recommendations that are from the ad-pool is captured by the bias probability, $\gamma$. We see that the probability of Type $II$ error decays exponentially with increasing $\gamma$, and also that larger $\gamma$ facilitates easier anomaly detection (Conditions~\ref{item:bias-prob}, \ref{item:num-eff-ads}).

\bigskip\noindent{\bf Choice of $c=1/2$.}
In the course of the proof of Theorem \ref{THM:MAIN}, we prove that 
for any choice of $c$, the  Type I Error is bounded by $O(\,Q(m)\,m^{-c}\,)$. 
Hence, by changing $c$ in the algorithm, one can control the error probability. 
However, the downside of increasing $c$ is that it effectively increases the threshold $T(t)$, 
which results in requiring the bias probability $\gamma=\omega(\log m (1+(c/A))/n)$. 

\subsection{Applications}
\label{subsec:applications}
The proposed anomaly detection algorithm is readily applicable in the retail market. 
It can identify recommender systems that dole out sponsored advertisements in the garb of personalized recommendations. In this era of personalization, there are numerous other applications, two of which are described below. These two examples also illustrate the advantage of \emph{BiAD} in requiring simple binary feedback, 
allowing it to be applied in a wide variety of scenarios.

\bigskip\noindent{\bf Search Engine Bias.}
Search engine bias is one of the most important ethical issues surrounding search engines, 
and its social implications have been studied for more than a decade \cite{introna00shaping-the-web, zimmer06implications-paid-search, epstein14google-democracy}.  
The Stanford Encyclopedia of Philosophy \cite{tavani14sep-ethics-search} describes 
search engine bias  as non-neutrality of search engines,  where ``search algorithms do not use objective criteria'' or ``favor some values/sites over others in generating their list of results for search queries.'' 
A sponsored search engine in the late 1990s called GoTo ranked 
its search results purely based on bids from advertisers \cite{elgesem08search-engine-bias}. 
It was evidently unsuccessful due to users' mistrust of paid searches and was eventually acquired by Google. Google also uses an auction   to sell ads but displays them physically separated from organic search results.

 The pros and cons of enforcing transparency in the algorithms used for generating search results have been examined in \cite{elgesem08search-engine-bias, granka10politics-of-search}. Even in the absence of total transparency, anomaly detection systems such as \emph{BiAD} could be useful in identifying bias in search engines. With personalization being extended to search results \cite{dou07personalized-search, google09personalized-search, bing11personalized-search}, search engines virtually act as recommendation engines. With large number of potential search results, our model of recommender systems with a large database fits well in this problem where biased search engines correspond to \emph{biased} recommender systems. 
 In addition to search engines, this example can be extended to identify hidden sponsored advertisements in social networking sites and online news portals, all of which use personalization algorithms.

\bigskip\noindent{\bf Pharmaceutical Lobby.}
Pharmaceutical lobbying is another controversial issue that affects many parts of the world \cite{abraham02pharma-industry-political, landers04pharma-lobbying-us, lancet05pharma-lobbying-uk, shah10pharma-lobbying}. 
Among its many aspects, 
we focus on the marketing practices of large pharmaceutical companies which manipulate the opinions of doctors, health care providers and law-makers by providing biased information and through other tactics \cite{blumenthal04doctors-drug-companies, fugh07doctors-drug-companies}. There have been allegations that big drug companies influence physicians to prescribe their highly priced branded drugs even when other better or cheaper alternatives are available \cite{landefeld09neurontin-marketing, goldacre14bad-pharma}.

Again, our interpretation of a \emph{biased} recommendation engine is well-suited to model this scenario. Since drugs are prescribed on a person-to-person basis, health care providers can be viewed as recommendation service providers who recommend drugs to patients, and the lobbying drug companies act as advertisers. A health care system that favors a few incompetent (expensive or ineffective) drugs in spite of the availability of other \emph{better} (cheaper or more effective) alternatives matches well with our definition of a \emph{biased} recommendation engine. With data samples consisting of prescriptions and their efficacy on patients, anomaly detection algorithms like \emph{BiAD} could help watchdog agencies in identifying such malpractices.

\section{Numerical Results}
\label{sec:num-results}
We evaluate our algorithm through offline simulations, with careful considerations for ensuring proximity to real world scenarios.
\subsection{Simulation Setup}
\label{subsec:setup}
Given below is a detailed description of the methods we adopt in our simulations to replicate the different components of a recommender system.

\bigskip\noindent{\bf User-Item Database}.
Estimating users' opinions about all items in a database is essential for real-world recommender systems, 
and hence for simulating those recommender systems as well. 
However, the ground truth on such data set is not available, 
since in existing data sets each user typically only rates a small subset of items, and those ratings are also noisy and possibly biased.  

For a complete user-item rating matrix, we take available sparse data sets (\cite{mcauley2013hidden, bennett2007netflix,glens}) and renormalize the ratings on a linear scale from zero to ten. 
The missing entries in the sparse matrix are then filled in using the matrix completion algorithm from \cite{montanarimatrix}. 
For the purpose of simulating real-world user opinions, we consider this completed matrix as ground truth. 
We evaluate our algorithm on three data sets:
\begin{enumerate}[label=D\arabic{*}, topsep=0pt,itemsep=-1ex,partopsep=1ex,parsep=1ex]
\item \label{item:D1} a subset of the Amazon cellphones and accessories data set \cite{mcauley2013hidden} with $3671$ users and $8728$ items,
\item \label{item:D2} a subset of the Netflix Prize data set \cite{bennett2007netflix} with $2951$ users and $9259$ movies, and
\item \label{item:D3} a subset of the Movielens 10m ratings data set \cite{glens} with $3671$ users and $8729$ movies.
\end{enumerate}

\bigskip\noindent{\bf Recommendation Engine}.
Due to non-transparency of recommendation strategies, it is not exactly known how recommendation engines behave. As a representation of the learning strategies used by these systems, we use two learning algorithms popular in literature:
\begin{enumerate}[label=L\arabic{*}, topsep=0pt,itemsep=-1ex,partopsep=1ex,parsep=1ex]
\item \label{item:L1} Matrix factorization. Specifically, we use the inexact ALM method proposed in \cite{lin2010augmented}.
\item \label{item:L2} User-based collaborative filtering (with Pearson correlation as the similarity metric \cite{resnick1994grouplens}).
\end{enumerate}

To simulate the temporal dynamics of a recommender system, the recommendation engine is initially supplied a sparse subset of the user-item ratings chosen according to a power-law degree distribution observed in real-world data sets \cite{huang2007analyzing}. (Specifically, the number of feedback entries from each user is chosen from a $pareto(3,3)$ distribution.) In each round of recommendation, the engine recommends one item to each user and observes users' feedback about the recommended items. It periodically updates its estimate of the users' preferences based on this feedback. In our experiments, we set the frequency of these updates to once in every $5$ rounds. 

It is natural for a recommendation engine to have an \emph{explore} component to address the cold start problem and have wider coverage of the database~\cite{li2010contextual}. Therefore, in all our experiments we invoke random explore for $0.1$ fraction of the recommendations made. In a recommendation meant for exploration, an item is chosen uniformly at random from the database, provided it has not been shown previously. 
 
For all other recommendations which do not explore, we use the following recommendation strategy -- for each user, the items are ranked according to the estimated preferences. To make a recommendation, an \emph{objective} recommendation engine chooses the highest ranked item among the items not yet recommended. The recommendation strategy of a \emph{biased} engine follows the description in Section~\ref{subsubsec:biased} -- for any recommendation, with probability $1-\gamma,$ like an \emph{objective} engine, it recommends the highest ranked item and with probability $\gamma,$ it recommends an item from the ad-pool. We consider two kinds of ad selection strategies -- from the among the ads that have not already been recommended to the user,
\begin{enumerate}[label=A\arabic{*}, topsep=0pt,itemsep=-1ex,partopsep=1ex,parsep=1ex]
\item \label{item:A1} An ad item is chosen uniformly at random.
\item \label{item:A2} The ad item which has the highest ranking is chosen. 
\end{enumerate}
 Strategy~\ref{item:A2} corresponds to customization of ads according to users' tastes. Note that this is harder to detect than strategy~\ref{item:A1} since it has higher likelihood of recommending effective ads.

\bigskip\noindent{\bf Anomaly Detection System}.
For players in the anomaly detection system, we randomly choose a subset of users from the data set. In experiments which test the performance of the algorithm with increasing number of players, we choose the subset of players incrementally. 

As explained in Sections~\ref{sec:sys-model} and \ref{sec:theory-results}, the algorithm requires feedback samples of efficacy of the recommendations made to the players. For our experiments, we adopt the characterization of efficacy used in our theoretical model given by Definition~\ref{defn:effective}. Note that the number of effective items in the database for any user depends on the efficacy threshold $\eta.$ To be able to test the algorithm for different number of effective items, we choose a different efficacy threshold for each of the data sets. Specifically, we set $\eta = 5.5, 8.0,$ and $8.8$ for data sets~\ref{item:D1}, \ref{item:D2}, and \ref{item:D3}, which correspond to an average number of effective items of $80,$ $250,$ and $150$ respectively.

\subsection{Results}
\label{subsec:results}
We evaluate the performance of \emph{BiAD} with variations in different parameters of the recommender system and the anomaly detection system. To demonstrate its effectiveness in different settings, we present performance results for various combinations of data sets (\ref{item:D1}-\ref{item:D3}) and recommendation algorithms (\ref{item:L1}-\ref{item:L2}, \ref{item:A1}-\ref{item:A2}). An \emph{objective} recommendation engine is represented by its learning algorithm (\ref{item:L1} or \ref{item:L2}) while a \emph{biased} recommendation engine is represented by its learning algorithm (\ref{item:L1} or \ref{item:L2}) and its ad-recommendation strategy (\ref{item:A1} or \ref{item:A2}). 
Although we present experimental results for specific combinations for space limitations, 
other settings give similar trade offs. 
Specifically, the simulation results corroborate our theoretical analysis of the tradeoffs between various parameters in Section~\ref{sec:discussion}.

We now describe how the performance depends on the choice of various parameters in \emph{BiAD}. 
We set the parameter $\hat{A}(t)$ according to Equation~\eqref{eqn:hatA} in all the simulations. Other parameters of the algorithm are discussed below.

\bigskip\noindent{\bf Threshold}.
As explained in Section~\ref{subsec:thresh}, varying the threshold parameter $T(t)$ in the algorithm affects Type $I$ and Type $II$ error probabilities in opposite ways. Using lower values of threshold increases probability of Type $I$ error and decreases that of Type $II$ error. The threshold given by Equation~\eqref{eqn:T} is designed to ensure, irrespective of the number of players in the anomaly detection system, low probability of false positive (Type $I$ error) even when the estimated user preferences are noisy. This is especially important if the risk associated with false implication of an \emph{objective} recommendation engine is high.

We observe that a less conservative threshold gives a better balance between the two types of errors. Specifically, we use a threshold that can be proved to guarantee low error rates under the assumption that the recommendation engine's estimation of user preferences are accurate. This threshold, denoted by $T'(t)$, is equal to the value of $\hat{p}(t)$ given by Equation~\eqref{eqn:hatp}. In all our simulations, we show the performance of \emph{BiAD} for both these threshold choices. Simulation results show that $T'(t)$ gives better performance than $T(t)$ except in one case (Figure.~\ref{subfig:numuser2}) where those two choices give similar performances. 
 
In both these thresholds, $\EE [ P^{\hat{\mathcal{A}}}_u(l) ]$ in Equation~\eqref{eqn:p} is substituted with the right hand side of \eqref{eqn:approx-thresh}. This requires knowledge of the number of effective items for each player in the anomaly detection system. In our simulations, \emph{BiAD} approximates this with the average number of effective items for all the users. Effectively, it uses the following value of $p(t)$ instead of Equation~\eqref{eqn:p}: 
\begin{align}
\label{eqn:est-eff-items}
p(t) &= \sum_{l=1}^t \frac{n \hat{A}(t)}{\tilde{f}([m]) - l + 1},
\end{align}
where $\tilde{f}([m])$ is an estimate of the average number of effective items in the database $[m]$. We assume that this average number is not very difficult to estimate and in all the results, unless specified, \emph{BiAD} has an accurate estimate of this number.

\bigskip\noindent{\bf Number of Rounds in the Test}. 
The number of rounds of recommendation $Q(m)$ affects the error probability. This is seen in Figure~\ref{fig:qm} which shows the variation of sum of Type $I$ and Type $II$ errors with $Q(m)$. Type $I$ error rate is in fact close to zero for all values of $Q(m)$ for both the thresholds, so the plots effectively show Type $II$ error rates. Theorem~\ref{THM:MAIN} guarantees detection of a \emph{biased} engine if $Q(m) \geq A.$ The plots show that \emph{BiAD} detects $8$ ads if $Q(m)$ is at least $8$ and $15$ for $T'(t)$ and $T(t)$ respectively.
\begin{figure}
\captionsetup{width=0.4\textwidth}
\begin{center}
\subfloat[Data set : \ref{item:D1} , Algorithm : \ref{item:L2} + \ref{item:A1} , $n = 100, \gamma = 0.45, A = 8.$]{\includegraphics[width=.5\linewidth]{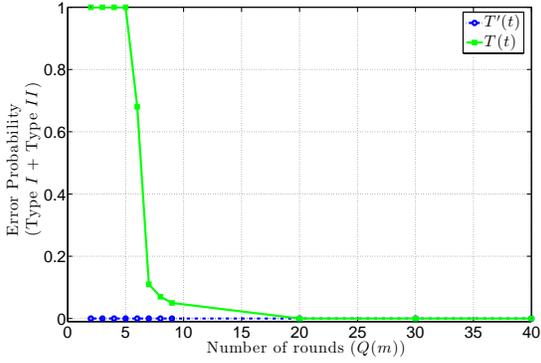}}
\hfill
\subfloat[Data set : \ref{item:D2} , Algorithm : \ref{item:L2} + \ref{item:A2} , $n = 100, \gamma = 0.35, A = 8.$]{\includegraphics[width=.5\linewidth]{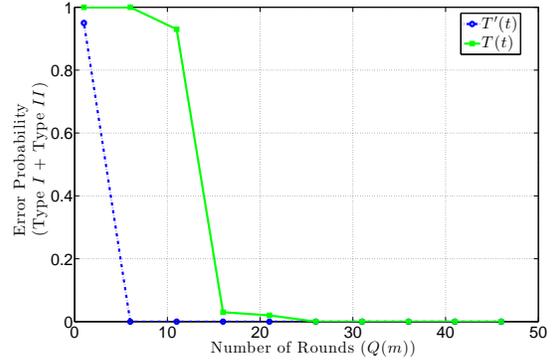}}
\captionsetup{width=\textwidth}
\caption{Number of rounds in the test, $Q(m)$ affects the number of ads that can be detected at least $8$ and $15$ rounds required for $T'(t)$ and $T(t)$ respectively.
}
\label{fig:qm}
\end{center}
\end{figure}
For all the remaining simulations, we set the parameter $Q(m) = 40.$

\bigskip\noindent{\bf Number of Players}.
Larger number of players in the anomaly detection system indicates higher number of input samples to the algorithm, and as expected, the algorithm performs better as this number increases. In Figure~\ref{fig:numuser}, we plot the sum of Type $I$ and Type $II$ error rates with increasing number of users. To detect a \emph{biased} engine with the specified value of $\gamma$, these plots show that $70$ and $100$ players respectively are sufficient when $T'(t)$ and $T(t)$ are chosen to be the threshold parameter. We use $100$ players in all other simulations.
\begin{figure}
\captionsetup{width=0.4\textwidth}
\begin{center}
\subfloat[Data set : \ref{item:D1}, Algorithm : \ref{item:L1} + \ref{item:A1}, $A = 8, \gamma = 0.45.$]{\includegraphics[width=.5\linewidth]{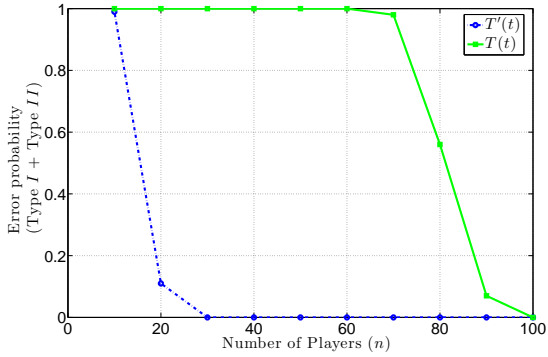}}
\hfill
\subfloat[Data set : \ref{item:D2}, Algorithm : \ref{item:L1} + \ref{item:A1} , $A = 8, \gamma = 0.35.$ \label{subfig:numuser2}]{\includegraphics[width=.5\linewidth]{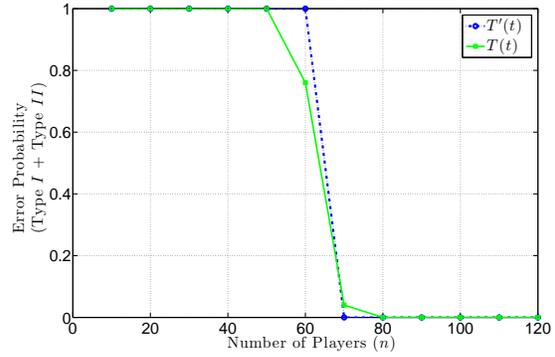}}
\captionsetup{width=\textwidth}
\caption{The performance improves with  number of collaborating users $n$. }
\label{fig:numuser}
\end{center}
\end{figure}

In addition to the choice of parameters in the algorithm, various aspects of the recommender system affect the performance of \emph{BiAD}. These are described below.

\bigskip\noindent{\bf Size of the Database}. 
Theorem~\ref{THM:MAIN} shows that \emph{BiAD} performs well for recommender systems with large item databases. Databases of varying size are constructed by sub-sampling items from the original data set. Figure~\ref{fig:numplayer} shows the variation of Type $I$ and Type $II$ errors with the size of the database. $T(t)$ and $T'(t)$ have very similar performance for the parameters in this experiment. The plot shows that, for detection of $8$ ads recommended $35$ percent of the time, the algorithm is effective for databases of size $1500$ items or larger.

\begin{figure}
\begin{center}
{\includegraphics[width=.5\linewidth]{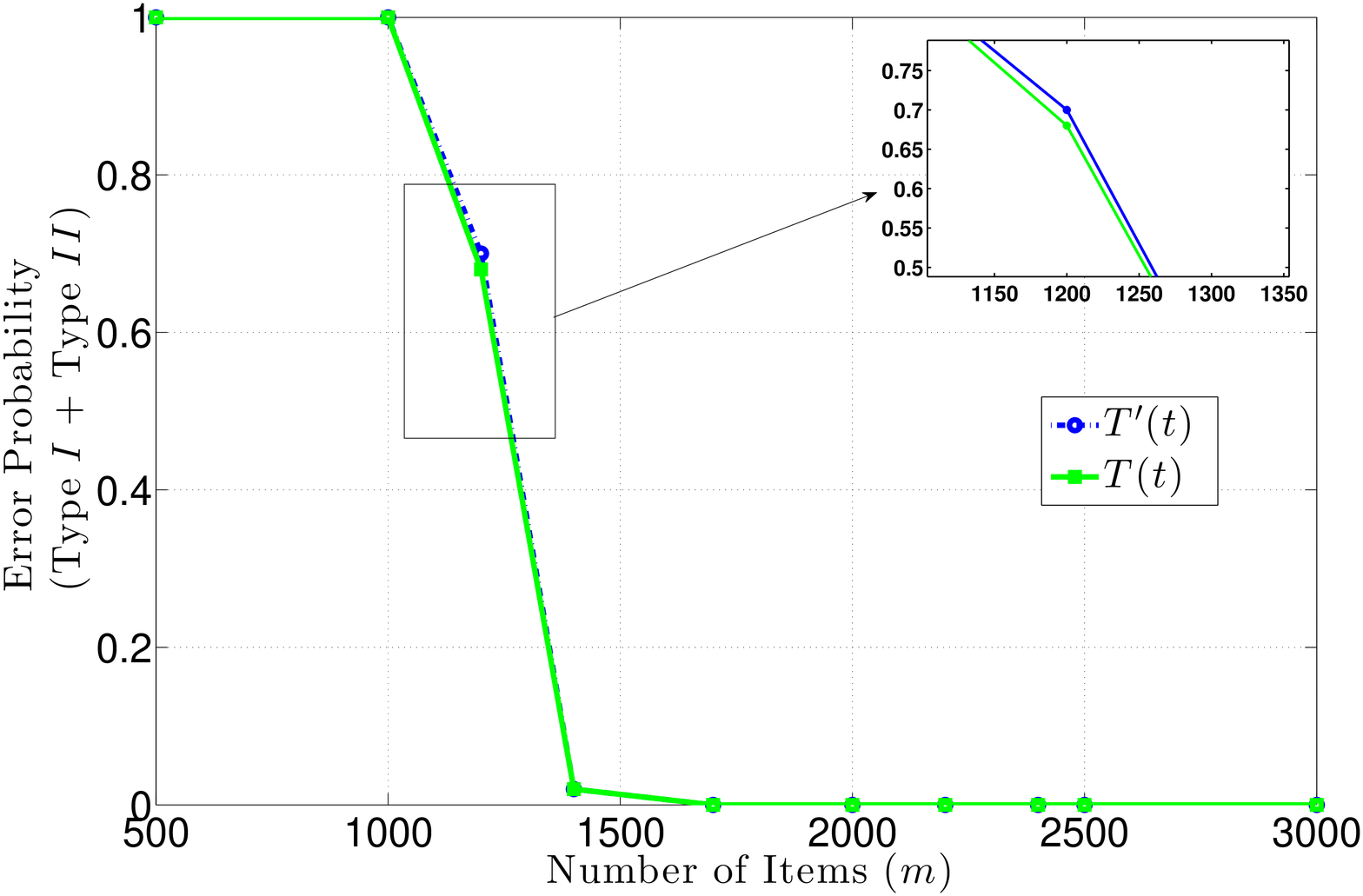}}
\caption{Variation of Type $I$ + Type $II$ error rates with size of data set. Data set : \ref{item:D2}, Algorithm : \ref{item:L1} + \ref{item:A2}, $A = 8, \gamma = 0.35, n = 100.$.
When there are more choices to recommend, the user satisfaction with {\em objective} recommender systems improves making  detection easier. 
}
\label{fig:numplayer}
\end{center}
\end{figure}
We now demonstrate that \emph{BiAD} has been appropriately designed to identify \emph{biased} engines that systematically recommend (make a sizable fraction of recommendations) from a small ad-pool.

\bigskip\noindent{\bf Size of the Ad-Pool}.
Theorem~\ref{THM:MAIN} shows that \emph{BiAD} guarantees detection of a \emph{biased} engine that has a small ad-pool. This same effect is also observed in simulations -- Figure~\ref{fig:numads} shows rate of missed detection (Type $II$ error rate) with varying size of the ad-pool. It is seen that both the thresholds perform well for small number of ads, while threshold $T'(t)$ can detect an ad-pool of size upto $25$.
\begin{figure}
\captionsetup{width=0.4\textwidth}
\begin{center}
\subfloat[Data set : \ref{item:D1} , Algorithm : \ref{item:L1} + \ref{item:A1} , $n = 100, \gamma = 0.45.$]{\includegraphics[width=.5\linewidth]{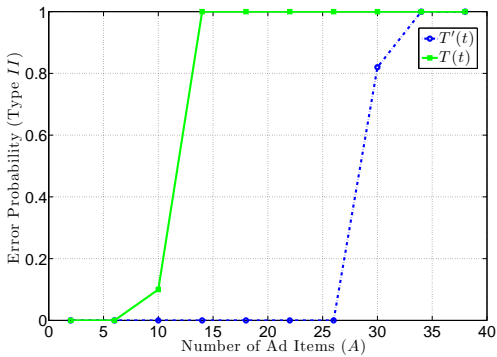}}
\hfill
\subfloat[Data set : \ref{item:D2} , Algorithm : \ref{item:L2} + \ref{item:A2} , $n = 100, \gamma = 0.35.$]{\includegraphics[width=.5\linewidth]{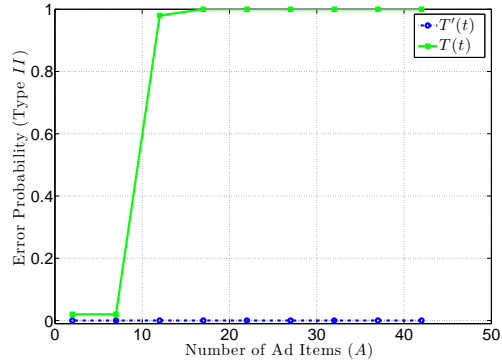}}
\captionsetup{width=\textwidth}
\caption{As the size of the ad-pool $A$ increases, the (personalized) ads become similar to effective recommendations, making 
it hard to detect (Type $II$ error is large).}
\label{fig:numads}
\end{center}
\end{figure}

\bigskip\noindent{\bf Bias Probability}.
The bias probability $\gamma$ quantifies the intensity of bias of the recommendation engine. Plots (Figure~\ref{fig:biasp1}) for Type $II$ error rate with $\gamma$ show that \emph{more biased} (higher $\gamma$) engines are easier to detect. 
\begin{figure}
\captionsetup{width=0.4\textwidth}
\begin{center}
\subfloat[Data set : \ref{item:D2} , Algorithm : \ref{item:L1} + \ref{item:A2} , $n = 100, A = 8.$]{\includegraphics[width=.5\linewidth]{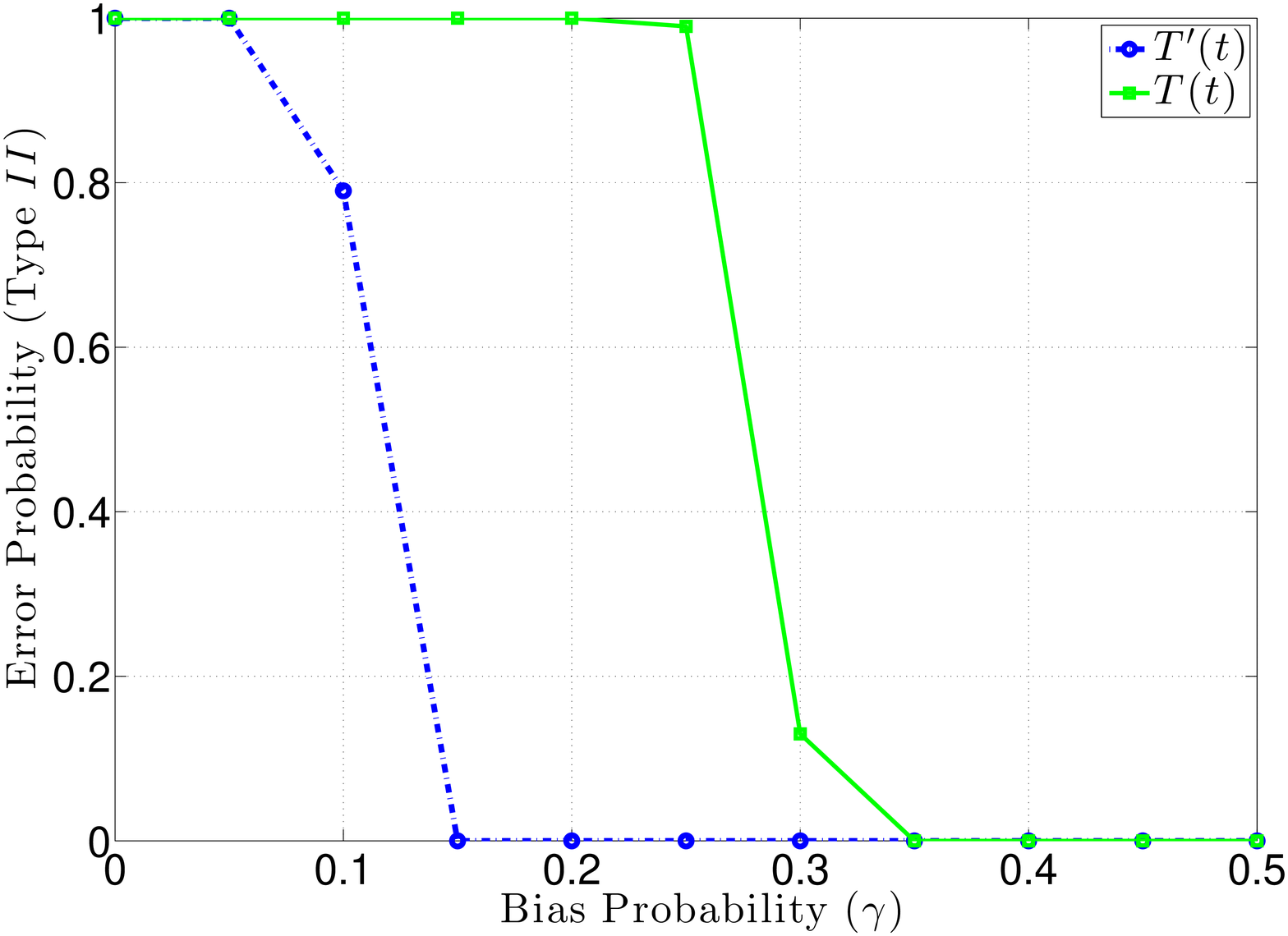}}
\hfill
\subfloat[Data set : \ref{item:D3} , Algorithm : \ref{item:L1} + \ref{item:A2} , $n = 100, A = 8.$]{\includegraphics[width=.5\linewidth]{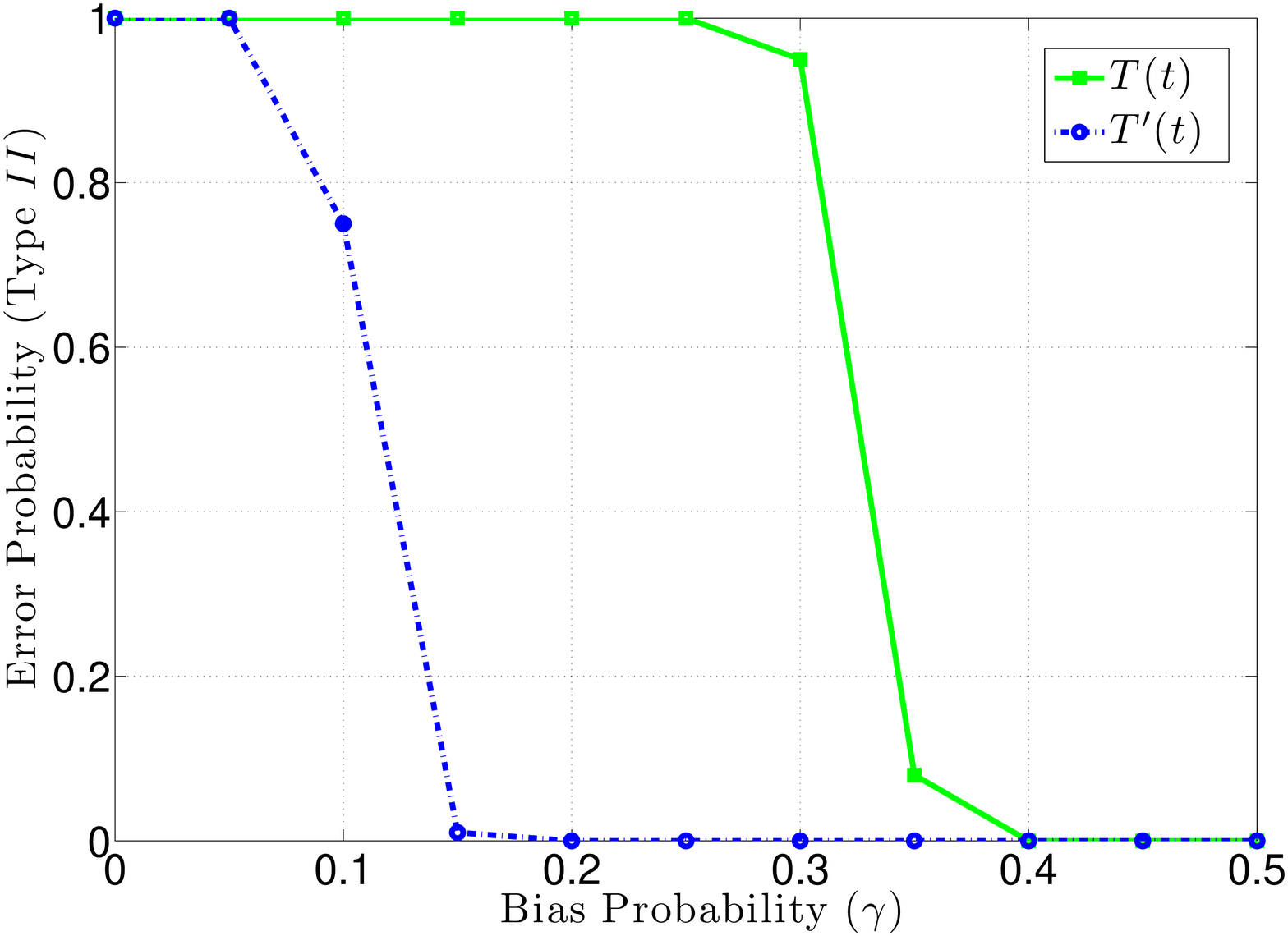}}
\captionsetup{width=\textwidth}
\caption{Type $II$ error rate decreases as bias probability $\gamma$ increases.}
\label{fig:biasp1}
\end{center}
\end{figure}

\bigskip\noindent{\bf Estimate of the Number of Effective Items}. 
In all the simulations above, it is assumed that \emph{BiAD} has an accurate estimate of the average number of effective items ($\tilde{f}([m])$) which is used to determine the threshold parameter in the algorithm (See Equation~\eqref{eqn:est-eff-items}). Note that overestimation of this parameter lowers the threshold parameter thereby increasing the probability of Type $I$ error and decreasing the probability of Type $II$ error. Figure~\ref{fig:fm} shows the effect of variations in this estimate for data set~\ref{item:D3} which has an average of $150$ effective items. We observe that $T'(t)$ performs well for a wide range of estimates. In the case of $T(t),$ it is safer to overestimate the parameter $\tilde{f}([m])$  than to underestimate it. 
\begin{figure}
\begin{center}
{\includegraphics[width=.5\linewidth]{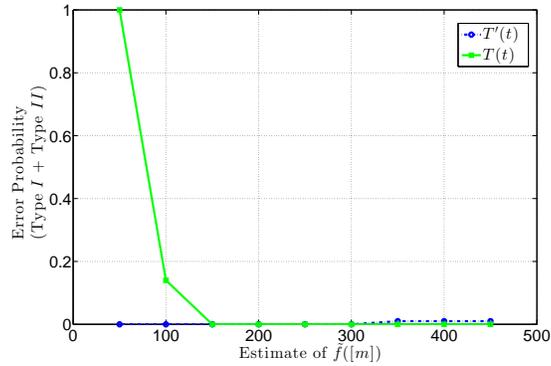}}
\caption{Variation of Type $I$ + Type $II$ error rates with perturbations in the algorithm's estimate of the average number of effective items $\tilde{f}([m])$. Data set : \ref{item:D3}, Algorithm : \ref{item:L1} + \ref{item:A1}, $A = 8, \gamma = 0.4, n = 100$.}
\label{fig:fm}
\end{center}
\end{figure}

\subsection{Ineffectiveness of Basic Average Test} 
\label{subsec:basic-avg}
As explained in the introduction (Section~\ref{sec:intro}), we demonstrate the inability of the basic average test to distinguish between random errors and deliberate promotion of ads. This test computes the average rating across all recommendations and decides between the two hypotheses based on a threshold parameter. With the specifics of the recommendation strategy (explore probability) unknown, it is difficult to estimate the right value of threshold. For an explore probability of $0.1,$ Figure~\ref{fig:naive1} shows the performance of the basic average test for different values of the threshold, denoted $\tau.$ It is seen that threshold values around $3$ give the best performance. But, as shown in Figure~\ref{fig:naive2}, this same threshold value fails for other values of explore probability. For example, the basic average test falsely declares an \emph{objective} recommendation engine with $20$ percent explore probability as \emph{biased}. This shows that the correct choice of $\tau$ is sensitive to the explore probability. In contrast, note that \emph{BiAD} has nearly zero Type $I$ error rate for all values of explore probabilities.

\begin{figure}

\centering
{\includegraphics[width=0.5\linewidth]{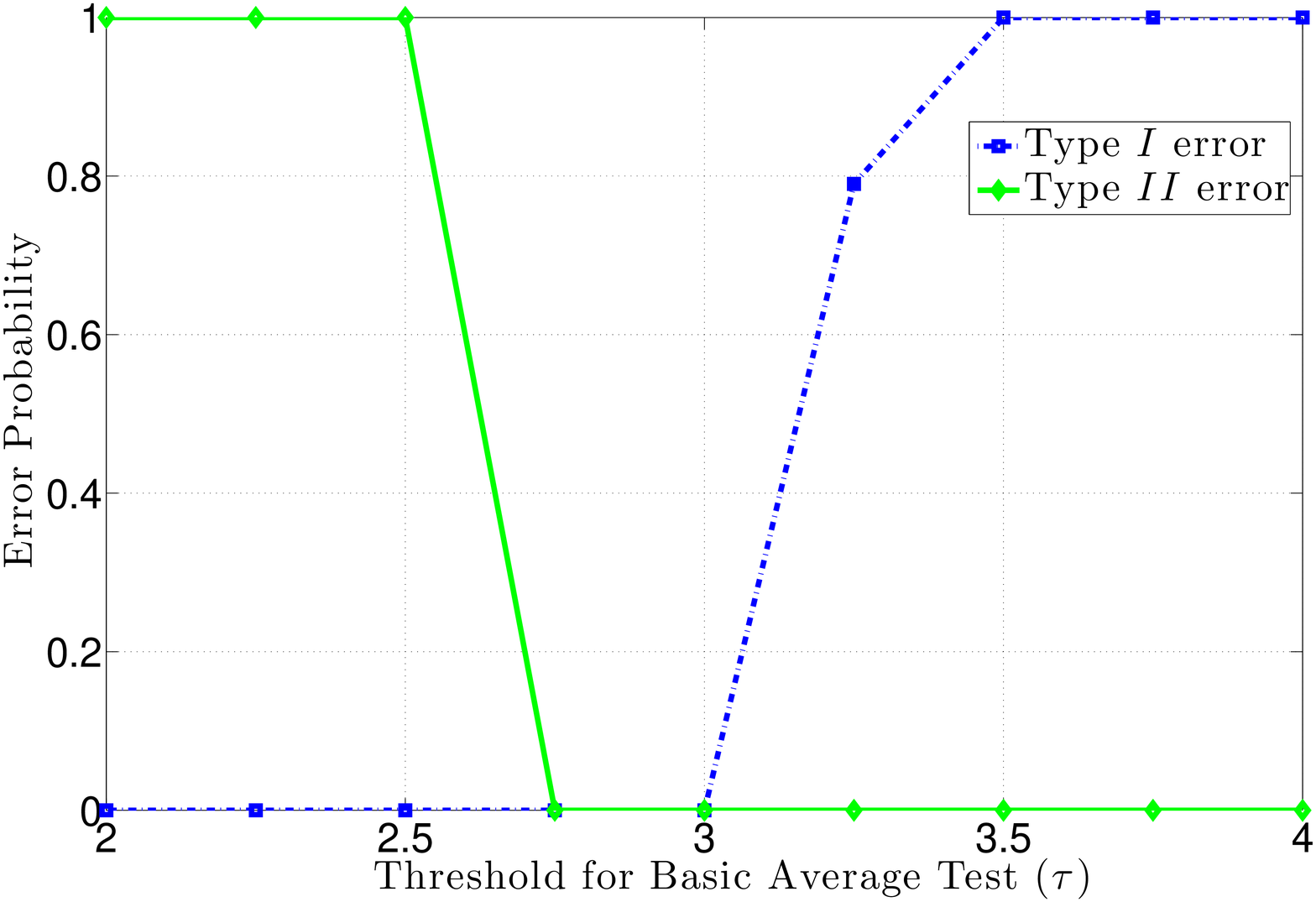}}
\caption{Variation of Type $I$ and Type $II$ error rates with threshold $\tau$ for the basic average test shows that 
the naive approach is sensitive to the choice of parameter $\tau.$
Data set : \ref{item:D1}, Algorithm : \ref{item:L1} + \ref{item:A1}, $A = 8, \gamma = 0.45, n = 100$, Explore Probability = $0.1.$}\label{fig:naive1}

\centering
{\includegraphics[width=0.5\linewidth]{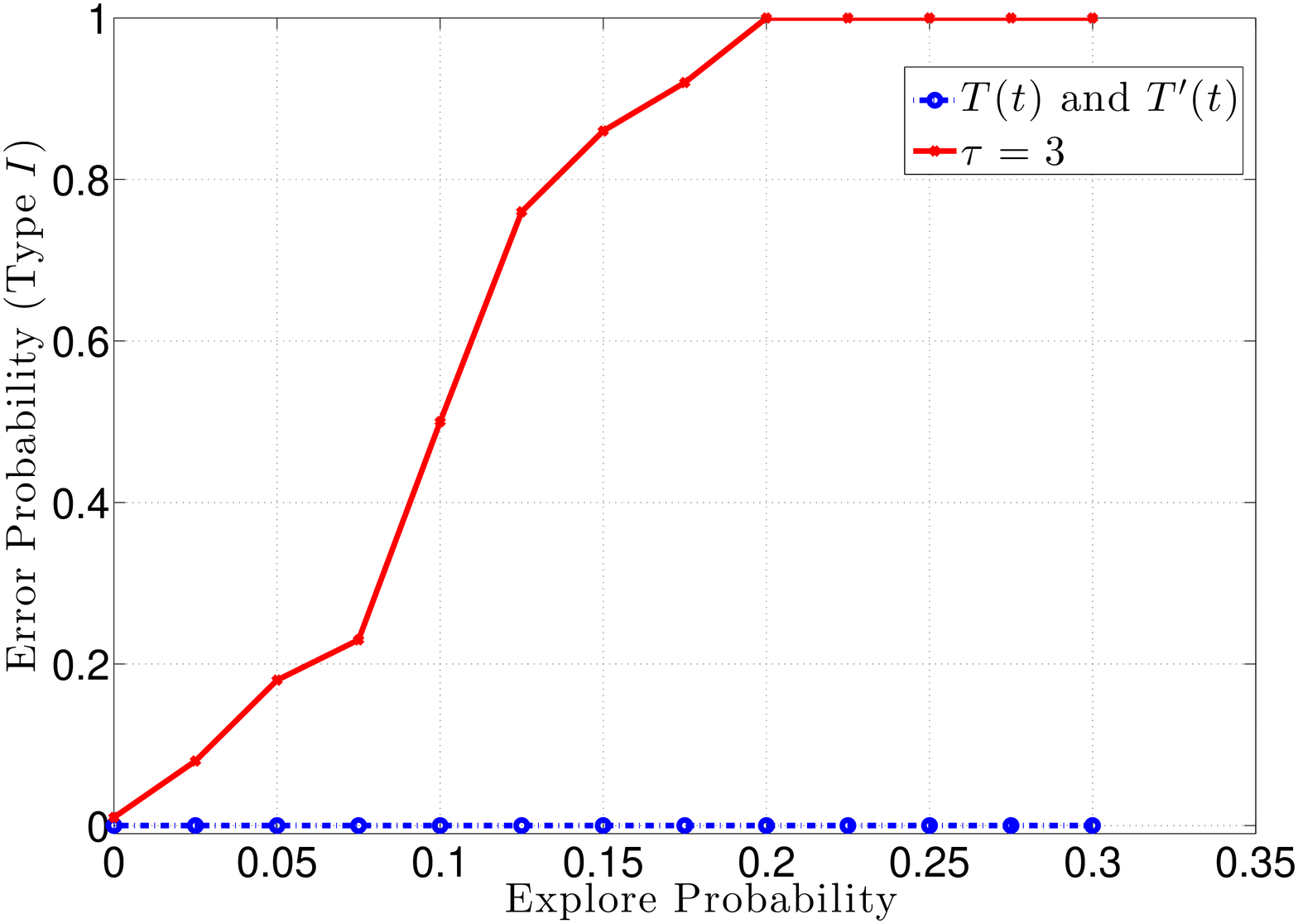}}
\caption{Variation of Type $I$ error rates with variation in explore probability shows that the threshold for the basic average test is sensitive the value of explore probability. Data set : \ref{item:D1}, Algorithm : \ref{item:L1},  $n = 100.$}
\label{fig:naive2}

\end{figure}

\section{Proofs}
\label{sec:proof}

Before we provide the proof of the main theorem, 
we state and prove two lemmas used in proving the main theorem.
\begin{lemma}
\label{lem:weights-upper-bound}
	For any \emph{objective} recommendation algorithm 
	and for any user $u$, item $i$, and time $t<(F(\hat{R}_{ui}(t), \hat{\bR}_u(t))+1)$, 
	the probability that item $i$ is recommended to user $u$ at time $t$ is upper bounded by 
\begin{align}
	W_{ui}(t) \leq \frac{1}{F\left(\hat{R}_{ui}(t), \hat{\bR}_u(t)\right) - t + 1} \;. 
	\label{eqn:lemma2}
\end{align}
\end{lemma}
\begin{proof}
\begin{align*}
1 & \stackrel{(a)}{=} \sum_{j=1}^m W_{uj}(t)	\\
& \geq \sum_{j=1}^m \mathds{1}\left\{\hat{R}_{uj}(t) \geq \hat{R}_{ui}(t)\right\}\left(1-\sum_{l=1}^{t-1} \mathds{1}_{ui}(l)\right) W_{uj}(t)	\\
& \stackrel{(b)}{\geq} \sum_{j=1}^m \mathds{1}\left\{\hat{R}_{uj}(t) \geq \hat{R}_{ui}(t)\right\}\left(1-\sum_{l=1}^{t-1} \mathds{1}_{ui}(l)\right) W_{ui}(t)	\\
& =  \left\vert\left\{j: \hat{R}_{uj}(t) \geq \hat{R}_{ui}(t), \, \sum_{l=1}^{t-1} \mathds{1}_{uj}(l) = 0\right\}\right\vert W_{ui}(t)	\\
& \geq \left( \left\vert\left\{j: \hat{R}_{uj}(t) \geq \hat{R}_{ui}(t)\right\}\right\vert - \left\vert\left\{j: \sum_{l=1}^{t-1} \mathds{1}_{uj}(l) \neq 0\right\}\right\vert \right) W_{ui}(t)	\\
& = \left(F\left(\hat{R}_{ui}(t), \hat{\bR}_u(t)\right) - (t - 1) \right) W_{ui}(t),
\end{align*}
where the $(a)$ and $(b)$ follow from the characterization of an \emph{objective} recommendation algorithm in Section~\ref{subsubsec:objective} -- equality $(a)$ due to the fact that $\bW(t)$ is a stochastic matrix (Property~\ref{item:stoc-mx}), and inequality $(b)$ due to the monotonic property satisfied by the weight matrix (Property~\ref{item:monotonic}). 
The above inequality gives the desired bound in \eqref{eqn:lemma2} for $t<(F(\hat{R}_{ui}(t), \hat{\bR}_u(t))+1)$.

\begin{lemma}
\label{lem:chernoff-bound}
Let $\{X_i, \, i=1, \dots, k\}$ be independent Bernoulli random variables with mean $\{p_i, \, i=1, \dots, k\},$ and let $\sum_{i=1}^k p_i \leq p.$ Then,
\begin{align*}
\PP\left[\sum_{i=1}^k X_i \geq T\right] \leq \exp\left(-T \log\left(\frac{T}{p}\right) + T - p\right) \fall T > p.
\end{align*}
\end{lemma}

Using Chernoff bound for independent random variables, we have, for any $\theta > 0,$
\begin{align*}
\PP\left[\sum_{i=1}^k X_i \geq T\right] & \leq e^{-\theta T}\prod_{i=1}^k \EE\left[e^{\theta X_i}\right]	\\
& = e^{-\theta T}\prod_{i=1}^k \left(p_ie^{\theta} + 1 - p_i\right)	\\
& \leq e^{-\theta T} \left(\frac{1}{k}\sum_{i=1}^k \left(p_ie^{\theta} + 1 - p_i\right)\right)^k	\\
& \leq e^{-\theta T} \left(\frac{p}{k}\left(e^{\theta}-1\right) + 1\right)^k,
\end{align*}
where the second inequality follows from the fact that the geometric mean of non-negative numbers is at most their arithmetic mean, 
and the last inequality follows for $\theta>0$, which is true for the choice of 
$\theta = \log\left({(k-p)T}/({p(k-T)})\right)$ for any $T > p$. Then, we get 
\begin{align*}
\PP\left[\sum_{i=1}^k X_i \geq T\right] & \leq \left(\frac{p(k-T)}{(k-p)T}\right)^T\left(\frac{(k-p)T}{(k-T)k} + 1 - \frac{p}{k}\right)^k	\\
& = \left(\frac{p}{T}\right)^T \left(\frac{k-p}{k-T}\right)^{k-T}	\\
& = \left(\frac{p}{T}\right)^T \left(1 + \frac{T-p}{k-T}\right)^{k-T}	\\
& \leq \left(\frac{p}{T}\right)^T e^{ T - p}	\\
& = \exp\left(-T \log\left(\frac{T}{p}\right) + T - p\right).
\end{align*}
\end{proof}

\begin{proof}[Proof of Theorem~\ref{THM:MAIN}]
The proof of the theorem consists of two parts which give upper bounds for probability of Type I and Type II errors. 

\subsubsection*{Type I Error}
The algorithm makes a Type $I$ error if it declares an \emph{objective} recommendation engine to be \emph{biased}. This section shows that the probability that the algorithm makes Type $I$ Error is low.

Suppose that the recommendation engine uses an \emph{objective} recommendation algorithm. We first bound the probability that \emph{BiAD} accepts $H_1$ in round $t.$ Recall that the algorithm accepts $H_1$ in round $t$ if $S(t) \geq T(t),$ and that $$S(t) = \max_{\left\{\hat{\mathcal{A}} \subseteq [m]: \left\vert \hat{\mathcal{A}}\right\vert = \hat{A}(t) \right\}} \sum_{i \in \hat{\mathcal{A}}} B_i(t).$$ 

Consider a fixed $\hat{\mathcal{A}} \subseteq [m]$ such that $\vert \hat{\mathcal{A}} \vert = \hat{A}(t).$ We first bound the probability that $\sum_{i \in \hat{\mathcal{A}}} B_i(t) \geq T(t)$ and then use union bound over all possible $\hat{\mathcal{A}}$ to obtain an upper bound on the probability that $S(t) \geq T(t).$ In this direction, we define 
\begin{equation}
\label{eqn:dislike-count}
X_u(l) := \sum_{i \in \hat{\mathcal{A}}}  \mathds{1}_{ui}(l) \cdot \mathds{1} \left\{ R_{ui} < \eta \right\}.
\end{equation} 
Note that $X_u(l)$ is equal to $1$ if in round $l$, player $u$ is recommended some item from set $\hat{\mathcal{A}}$ that is not effective and is $0$ otherwise. Given $\left\{ \bW(l), \hat{\bR}(l), \, l \in [t] \right\},$ we have that $\left\{ X_u(l), \, l \in [t], \, u \in [n] \right\}$ are Bernoulli random variables independent of each other. For every $1 \leq u \leq n$ and $1 \leq l \leq t,$ the mean of $X_u(l)$ can be bounded as follows:
\begin{align}
\label{eqn:mean-dislike-count}
 \EE\left[ X_u(l) \bigg\vert \left\{ \bW(l'), \hat{\bR}(l')\right\}_{l' = 1}^t \right]
 & = \sum_{i \in \hat{\mathcal{A}}} W_{ui}(l) \mathds{1} \left\{ R_{ui} < \eta \right\}	\n	\\
& \leq P^{\hat{\mathcal{A}}}_u(l),
\end{align}
where the inequality follows from the upper bound for $W_{ui}(l)$ from Lemma~\ref{lem:weights-upper-bound} and the definition of $P^{\hat{\mathcal{A}}}_u(l)$ (\eqref{eqn:sum-prob}).
 From \eqref{eqn:p}, we have $$\sum_{u=1}^n \sum_{l=1}^t \EE\left[ P^{\hat{\mathcal{A}}}_u(l)  \right] \leq p(t).$$ Note that, since the noise in the rating estimates are independent across time and users, $\left\{ P^{\hat{\mathcal{A}}}_u(l), \, l \in [t], \, u \in [n] \right\}$ are independent random variables. We now use the Chernoff bound in Lemma~\ref{lem:chernoff-bound} to obtain a probabilistic upper bound on their sum.
\begin{align*}
\PP\left[ \sum_{u=1}^n \sum_{l=1}^t P^{\hat{\mathcal{A}}}_u(l) \geq \hat{p}(t) \right]
& \leq \exp\left(-\hat{p}(t) \left(\log\left(\frac{\hat{p}(t)}{p(t)}\right) - 1 \right)- p(t)\right).
\end{align*}
By the definition of Lambert-W function,
\begin{align*}
\hat{p}(t) \left(\log\left(\frac{\hat{p}(t)}{p(t)}\right) - 1 \right) = \left( \hat{A}(t) +c \right)\log m - p(t),
\end{align*}
 which further implies that 
\begin{align}
\label{eqn:prob-ub}
\PP\left[ \sum_{u=1}^n \sum_{l=1}^t P^{\hat{\mathcal{A}}}_u(l) \geq \hat{p}(t) \right] & \leq \exp \left( -\left( \hat{A}(t)+c \right) \log m  \right). 
\end{align} 

We now proceed to obtain a probabilistic upper bound for the sum, $\sum_{u=1}^n \sum_{l=1}^t X_u(l).$ By inequality~\eqref{eqn:mean-dislike-count}, the sum of the corresponding means has the following upper bound:
\begin{align*}
\sum_{u=1}^n \sum_{l=1}^t \EE\left[ X_u(l) \bigg\vert \left\{ \bW(l'), \hat{\bR}(l')\right\}_{l' = 1}^t \right] \leq \sum_{u=1}^n \sum_{l=1}^t P^{\hat{\mathcal{A}}}_u(l).
\end{align*}
Since $\left\{ X_u(l), \, l \in [t], \, u \in [n] \right\}$ are independent Bernoulli random variables given $\left\{ \bW(l), \hat{\bR}(l), \, l \in [t] \right\},$ Lemma~\ref{lem:chernoff-bound} can be used as before. 
If $\sum_{u=1}^n \sum_{l=1}^t P^{\hat{\mathcal{A}}}_u(l) \leq \hat{p}(t),$ then  Lemma~\ref{lem:chernoff-bound} gives
\begin{align}
\PP\left[ \sum_{u=1}^n \sum_{l=1}^t X_u(l) \geq T(t) \bigg\vert \left\{ \bW(l'), \hat{\bR}(l')\right\}_{l' = 1}^t \right]
& \leq \exp \left( -\left( \hat{A}(t)+c \right) \log m  \right).
\end{align}
The above upper bound can be derived in exactly the same manner as the upper bound in \eqref{eqn:prob-ub}. Combining this upper bound with inequality~\eqref{eqn:prob-ub}, we have
\begin{align}
\PP\left[ \sum_{u=1}^n \sum_{l=1}^t X_u(l) \geq T(t) \right]	
& \leq \PP\left[ \sum_{u=1}^n \sum_{l=1}^t X_u(l) \geq T(t) \bigg\vert \sum_{u=1}^n \sum_{l=1}^t P^{\hat{\mathcal{A}}}_u(l) \leq \hat{p}(t) \right]	\n + \PP\left[ \sum_{u=1}^n \sum_{l=1}^t P^{\hat{\mathcal{A}}}_u(l) \geq \hat{p}(t) \right]	\n	\\
& \leq 2\exp \left( -\left( \hat{A}(t)+c \right) \log m  \right).	\label{eqn:prob-ub2}
\end{align}
Now, recall that $B_i(t)$ is the number of players who have rated item $i$ ineffective upto round $t,$ which can be mathematically written as $B_i(t) = \sum_{u=1}^n \sum_{l=1}^t \mathds{1}_{ui}(l) \cdot \mathds{1} \left\{ R_{ui} < \eta \right\}.$ Using definition~\eqref{eqn:dislike-count}, we have $$\sum_{i \in \hat{\mathcal{A}}} B_i(t) =  \sum_{u=1}^n \sum_{l=1}^t X_u(l),$$ which gives us the following equivalent form of inequality~\eqref{eqn:prob-ub2}:
\begin{align}
\label{eqn:prob-ub3}
 \PP\left[ \sum_{i \in \hat{\mathcal{A}}} B_i(t) \geq T(t) \right] & \leq 2\exp \left( -\left( \hat{A}(t)+c \right) \log m  \right).	
\end{align}

We can now take a union bound over all possible $\hat{\mathcal{A}}$ to bound the probability that \emph{BiAD} accepts $H_1$ in round $t.$
\begin{align*}
\PP\left[S(t) \geq T(t)\right] 
& = \PP\left[\bigcup_{\left\{\hat{\mathcal{A}} \subseteq [m]: \left\vert \hat{\mathcal{A}}\right\vert = \hat{A}(t) \right\}} \left\{\sum_{i \in \hat{\mathcal{A}}} B_i(t) \geq T(t) \right\} \right]	\\
& \leq \sum_{\left\{\hat{\mathcal{A}} \subseteq [m]: \left\vert \hat{\mathcal{A}}\right\vert = \hat{A}(t) \right\}} \PP\left[ \sum_{i \in \hat{\mathcal{A}}} B_i(t) \geq T(t) \right]	\\
& \leq 2\binom {m}{\hat{A}(t)} \exp \left( -\left( \hat{A}(t)+c \right) \log m  \right)	\\
& \leq 2\exp \left( -c\log m \right)	\\
& = 2m^{-c},
\end{align*}
where the second inequality follows from \eqref{eqn:prob-ub3}. Further taking a union bound over all rounds,
\begin{align*}
\PP[\text{Type $I$ Error}] &= \PP\left[\cup_{t=1}^{Q(m)} S(t) \geq T(t) \right]	\\
& \leq \sum_{t=1}^{Q(m)} \PP\left[S(t) \geq T(t) \right]	\\
& \leq 2\,Q(m)\,m^{-c}.
\end{align*}
 This shows that \emph{BiAD} declares an \emph{objective} recommendation engine as \emph{biased} with probability $O(\frac{Q(m)}{\sqrt{m}})$ for the choice of $c=1/2$.

\subsubsection*{Type II Error}
The algorithm makes a Type $II$ error if it does not detect an \emph{biased} recommendation engine, i.e., it declares $H_0$ when $H_1$ is true. 
Suppose that the recommendation engine is \emph{biased} with an ad-pool $\mathcal{A}$ of size $\left\vert \mathcal{A} \right\vert = A.$ Fix a $\delta \in (0,1).$ We prove that $\sum_{u=1}^n \sum_{l=1}^A \sum_{i \in \mathcal{A}} \mathds{1}_{ui}(l),$ which is equal to the number of ad recommendations to the $n$ players until round $A$ is at least $\gamma(1-\delta) nA$ with high probability. For any $1 \leq u \leq n,$  let $Y_u(l) = 1$ if the \emph{biased} recommendation engine decides to recommend from the ad-pool to user $u$ in round $t.$ Note that $\sum_{i \in \mathcal{A}} \mathds{1}_{ui}(l) \geq Y_u(l)$ and that $\left\{ Y_u(l), \, l \in [A], \, u \in [n] \right\}$ are i.i.d.\ Bernoulli random variables with mean $\gamma.$ This gives us
\begin{align}
\label{eqn:ad-conc}
\PP\left[ \sum_{u=1}^n \sum_{l=1}^A \sum_{i \in \mathcal{A}} \mathds{1}_{ui}(l) < \gamma(1-\delta) nA \right]
& \leq \PP\left[ \sum_{u=1}^n \sum_{l=1}^A Y_u(l) < \gamma(1-\delta) nA \right]		\n	\\
& \leq e^{\left( - \frac{\delta^2\gamma nA }{2}\right)},
\end{align}
where the last inequality follows from a version of Chernoff bound given in \cite{mitzenmacher-upfal05probability} for sum of i.i.d.\ Bernoulli random variables.

The detection algorithm makes the correct decision if in round $t,$ $S(t) \geq T(t)$ for some $t \leq Q(m).$ We show that $S(A) \geq T(A)$ with high probability. Since $A \leq Q(m),$ the algorithm makes the correct decision with high probability. 
Now, suppose that the number of ad recommendations to the $n$ players until round $A$ is at least $\gamma(1-\delta) nA.$ Since the total number of effective ads to the $n$ players is $o(\gamma nA),$ the total number of ineffective recommendations from the ad-pool until round $A$ is $\gamma \Omega(nA).$ Consequently,
\begin{align}
\label{eqn:ads-lb}
S(A) & = \max_{\left\{\hat{\mathcal{A}} \subseteq [m]: \left\vert \hat{\mathcal{A}}\right\vert = A \right\}} \sum_{i \in \hat{\mathcal{A}}} B_i(t)	\n	\\
& \geq \sum_{i \in \mathcal{A}} B_i(t)	\n	\\
& \geq \gamma(1-\delta) nA - o(\gamma nA)	\n	\\
& = \gamma \Omega(nA).
\end{align}
To prove that $T(A)$ does not exceed the right hand side of inequality~\eqref{eqn:ads-lb}, we consider the following cases:
\begin{enumerate}[label=(\roman{*}): ]
\item $\hat{\beta}(A) \geq e$	\\
Since $W(\cdot)$ is an increasing function in $[0, \infty),$ we have $W\left( \frac{\hat{\beta}(A)}{e} \right) \geq W(1) > \frac{1}{2}.$ Now,
\begin{align*}
T(A) & = \exp \left(1+W\left(\frac{\hat{\beta}(A)}{e}\right)\right)\hat{p}(A)	\\
& = \frac{\hat{\beta}(A)}{W\left(\frac{\hat{\beta}(A)}{e}\right)}\hat{p}(A)	\\
& \leq \frac{(A+c)\log m }{W(1)}	\\
& = \frac{\log m}{n} O(nA),
\end{align*}
where the second equality follows from the definition of the Lambert W function, and the last inequality follows by using the definition of $\hat{\beta}(A)$ given by \eqref{eqn:hatbeta}.
\item $\hat{\beta}(A) < e, \; \beta(A) \geq e$
\begin{align*}
\hat{p}(A) & = \exp \left(1+W\left(\frac{\beta(A)}{e}\right)\right)p(A)	\\
& = \frac{\beta(A)}{W\left(\frac{\beta(A)}{e}\right)}p(A)	\\
& \leq \frac{(A+c)\log m }{W(1)}.	\\
T(A) & = \exp \left(1+W\left(\frac{\hat{\beta}(A)}{e}\right)\right)\hat{p}(A)	\\
& \leq \exp \left(1+W(1)\right)\hat{p}(A)	\\
& = \frac{\log m}{n} O(nA).
\end{align*}
\item $\hat{\beta}(A) < e, \; \beta(A) < e$
\begin{align*}
T(A) & = \exp \left(1+W\left(\frac{\hat{\beta}(A)}{e}\right)\right)\hat{p}(A)	\\
& \leq \exp \left(1+W(1)\right)\hat{p}(A) 	\\
& \leq \exp \left(2+2W(1)\right)p(A)	\\
& = \frac{p(A)}{nA} O(nA).
\end{align*}
\end{enumerate}
Since $\gamma = \omega\left( \frac{\log m}{n} \right), \, \omega(\frac{p(A)}{nA}),$ combining the results from the above three cases with inequality~\eqref{eqn:ads-lb} gives that $S(A) \geq T(A)$ for $m$ large enough. Therefore, the algorithm declares the correct hypothesis in round $A$ if the number of ad recommendations to the $n$ players until round $A$ is at least $\gamma(1-\delta) nA.$ We can therefore use the concentration inequality in \eqref{eqn:ad-conc} to bound the probability of Type $II$ error.
\begin{align*}
\PP[\text{Type $II$ Error}] & \leq \PP\left[ S(A) < T(A) \right]	\\
& \leq \PP\left[ \sum_{u=1}^n \sum_{l=1}^A \sum_{i \in \mathcal{A}} \mathds{1}_{ui}(l) < \gamma(1-\delta) nA \right]	\\
& \leq \exp\left( - \frac{\delta^2}{2}\gamma nA \right)	\\
& = e^{-\Omega(\gamma n)}.
\end{align*}
This shows that the probability of Type $II$ error decays exponentially with the number of players and the bias probability.
\end{proof}

\section{Conclusion}
We propose an algorithm that can identify an \emph{biased} recommendation engine that systematically favors a few sponsored advertisements over other genuine recommendations. We formulate a probabilistic model for recommender systems and give theoretical guarantees for our detection algorithm based on this model. Specifically, we show that the probability of missed detection and false positives are low for recommender systems with large databases. We show through simulations that the algorithm performs well for many data sets and different types of recommendation algorithms. In an age when both personalization and advertising have become very prevalent, this kind of anomaly detection algorithm is relevant in a wide variety of scenarios. We demonstrate how our detection algorithm can be applied to problems such as identification of search engine bias and pharmaceutical lobbying. It would be interesting to investigate ways of deploying such an anomaly detection mechanism in practical settings.

\section*{Acknowledgments}
This work was supported by the NSF grant CNS-1320175 and the ARO grant W911NF-14-1-0387.

\bibliographystyle{ieeetran}
\bibliography{recosys}

\end{document}